\documentclass[11pt]{article}

\usepackage[cmex10]{amsmath}   
\usepackage{amsthm,amssymb,mathtools,url}
\usepackage{bm}
\usepackage{fullpage}
\usepackage[pdftex]{graphicx} % remove pdftex if you are not compiling to pdf
\graphicspath{{./figures/}} 
\DeclareGraphicsExtensions{.pdf,.png,.jpg}

\usepackage[caption=false]{subfig}
\usepackage{float}

\usepackage{url}
\usepackage{xcolor}

\usepackage{paralist}

\newcommand{\partT}[1]{\partial_T #1}
\newcommand{\partX}[1]{\partial_X #1}
\newcommand{\partY}[1]{\partial_Y #1}
\newcommand{\partXX}[1]{\partial_X^2 #1}
\newcommand{\partYY}[1]{\partial_Y^2 #1}
\newcommand{\partXXYY}[1]{\partial_X^2\partial_Y^2 #1}

\newcommand{\partXXXX}[1]{\partial_X^4 #1}
\newcommand{\partYYYY}[1]{\partial_Y^4 #1}
\newcommand{\partt}[1]{\partial_t #1}
\newcommand{\partx}[1]{\partial_x #1}
\newcommand{\party}[1]{\partial_y #1}
\newcommand{\partxx}[1]{\partial_x^2 #1}
\newcommand{\partyy}[1]{\partial_y^2 #1}
\newcommand{\partxxyy}[1]{\partial_x^2\partial_y^2 #1}

\newcommand{\partxxxx}[1]{\partial_x^4 #1}
\newcommand{\partyyyy}[1]{\partial_y^4 #1}

\newcommand{\blockcomment}[1]{}

\newtheorem{theorem}{Theorem}
\theoremstyle{definition}
\newtheorem{defn}{Definition}

\begin{document}

\begin{center}
{\bfseries \Large Phase-Diffusion Equations for the Anisotropic Complex Ginzburg-Landau Equation} \\
\end{center}

\vspace{0.2ex}

\begin{center}
{\large Derek Handwerk, Gerhard Dangelmayr, Iuliana Oprea, and Patrick D. Shipman }\vspace{0.2ex}
Department of Mathematics, Colorado State University, Ft. Collins, CO 80523-1874
\end{center}

%\footnote{Current address: Insert!}

\vspace{2ex}

\centerline{\today}

\vspace{2ex}

\begin{abstract}
\noindent  The anisotropic complex Ginzburg-Landau equation (ACGLE) describes slow modulations of patterns in anisotropic spatially extended systems near oscillatory (Hopf) instabilities with zero wavenumbers.  Traveling wave solutions to the ACGLE become unstable near Benjamin-Feir-Newell instabilities.  We determine two instability conditions in parameter space and study codimension-one (-two) bifurcations that occur if  one (two) of the conditions is (are) met.  We derive anisotropic Kuramoto-Sivashinsky-type equations that govern the phase of the complex solutions to the ACGLE  and generate solutions to the ACGLE from solutions of the phase equations.  
 \end{abstract}

\vspace{2ex}
\noindent
\begin{mbox}
\end{mbox}
\vspace{2ex}

%\noindent
%\begin{mbox}
%{\bf MSC:} 
%\end{mbox}
%\vspace{2ex}

\hrule
\vspace{1.1ex}

\section{Introduction}\label{sec:intro}  %%%%%%%%%%%%%%%%%%%%%%%%%%%%

Complex spatiotemporal patterns such as spatiotemporal chaos and defects, manifest themselves in spatially extended systems driven far from equilibrium. The problem of finding a general framework for the characterization of such disordered states, as well as the identification of instability mechanisms generating them, remain active areas of investigation in nonlinear dynamics.  While complex spatiotemporal patterns in nonequilibrium isotropic systems have been intensively studied, far less is known about complex spatiotemporal dynamics in anisotropic systems.  Yet, intriguing features of patterns in anisotropic media when driven out of equilibrium have been observed in a wide range of experimental studies.  Examples include electroconvection of nematic liquid crystals \cite{kramerpesch95}, surface nanopatterning by ion-beam erosion \cite{IBS14}, chemical waves in catalytic surface reactions \cite{Mak06,rose96}, epitaxial growth \cite{rost1995anisotropic}, sea ice melting \cite{SeaIce18}, and vegetation patterns \cite{Vegetation11}.  %inclined-layier convection \cite{Huepe04}, 

These experimental observations and numerical simulations demonstrate that anisotropy can lead to novel mechanisms and phenomena that are manifested only in anisotropic media. For instance, Rayleigh-B\'enard convection \cite{Bodenschatz2000,Egolf2000}, a prototype of spatiotemporal chaos in isotropic fluids, displays a paradigm of spatiotemporal chaos known as spiral defect chaos, while in anisotropic systems such as nematic electroconvection other mechanisms intervene, leading to new patterns such as the zigzag spatiotemporal chaos \cite{Dennin96}.  Similar patterns can arise from even this wide range underlying physical systems.  Analysis of mathematical models for these diverse systems results in a universal characterization of the similar patterns by their description using \textit{amplitude} and \textit{phase} equations.  

\textit{Amplitude equations} describe the slow modulation of the pattern near the threshold of instability. They can be derived via symmetry arguments or through multiple-scales analysis. The same amplitude equation can be derived from different underlying equations up to some unknown coefficients that determine the length and time scales, and the effect of the nonlinearity. The underlying equations can be used to determine the correct coefficients of the amplitude equation or they can be scaled out completely. The ability to completely remove the system-dependent coefficients from the equation demonstrates the universality of amplitude equations. 

The slow modulations of a complex amplitude satisfying an amplitude equation are generally governed by \emph{phase equations} which describe the extremely slow variation of the phase of the amplitude. When a system has more than one extended spatial direction, it and thus also the amplitude and phase equations, can be either isotropic or anisotropic.  Nematic liquid crystals \cite{dangelmayroprea2004,dangelmayr2008modulational,kramerpesch95}, ion bombardment \cite{harrisonbradley16}, and surface erosion and growth \cite{rost1995anisotropic} are examples of physical scenarios that can be described by anisotropic amplitude and phase equations. 

The amplitude equation of concern in this paper is the two-dimensional, anisotropic complex Ginzburg-Landau equation (ACGLE),  
\begin{equation}
  \partt{A} = \mu A + (1 + i\alpha_1)\partxx{A} + (1+i\alpha_2)\partyy{A} - (1+i\beta)|A|^2A,
  \label{eq:ACGLE}
\end{equation}
where $A$ is the complex amplitude and $\alpha_1,\alpha_2,\beta,\mu\in\mathbb{R}$ with $\mu>0$. For \(\alpha_1 = \alpha_2\), this equation is the isotropic complex Ginzburg-Landau equation (CGLE). We refer to (\ref{eq:ACGLE}) as the 1D CGLE if the $y$-dependent term is absent and one looks only for solutions $A(x,t)$.  

Since the spatially independent part of the complex Ginzburg Landau equation coincides with the normal form for a supercritical Hopf bifurcation \cite{GuckenheimerHolmes83}, both the 1D and 2D isotropic complex Ginzburg Landau equation and variants thereof, referred to as $\lambda-\omega$ systems as incepted by Kopell and Howard \cite{KopellHoward73}, have been studied as spatiotemporal model equations showing plane wave and front solutions in various settings including reaction diffusion systems \cite{IpsenEtAl97,Kuramoto1984,KuramotoTsuzuki75,ShenEtAl18} as well as predator-prey systems \cite{BennetSherrat19,SherrattEtAl09}. Extending these studies for the ACGLE (\ref{eq:ACGLE}) should provide new insights into the effect of ansiotropies on this kind of spatiotemporal dynamics. 

There has been much less research into the ACGLE compared to the isotropic CGLE. Some aspects of phase chaos were investigated in \cite{faller1998phase}, and new chevron-like, ordered defect solutions were reported in \cite{faller1998chevrons}. A study of a perturbed amplitude in the longwave case was performed in \cite{brown1993evolution}.  The authors of this paper used a perturbed ansatz of the form \(A(x,y,t) = \sqrt{\mu}(1+r(x,y,t))e^{-i\beta\mu t}\) that lacks any phase perturbation as considered by us (see Equation (\ref{eq:aCGL_ansatz_longwave})), which is crucial for the reduction of (\ref{eq:ACGLE}) to a phase equation.  A comprehensive analysis of long-wave and short-wave instabilities of traveling wave solutions of the ACGLE was performed in \cite{dangelmayr2008modulational}. This analysis was done in the context of modulational instabilities of traveling waves determined as solutions of systems of two or four globally coupled complex Ginzburg-Landau equations, which  are the amplitude equations corresponding to oscillatory instabilities with nonzero critical wavenumbers of a basic state of an anisotropic system \cite{schneider2007}. 

Due to the assumptions used in the derivation of phase equations from amplitude equations, it is important to consider the regions and circumstances in which they are valid. In the case of the 1D CGLE, the validity of phase equations has been established in \cite{haas2018modulation,melbourne2004phase,van2004phase} for different parameter regimes. In the parameter regime where a 1D, periodic traveling wave solution \(A(x,t) = A_qe^{i(qx + \omega_qt)}\) is stable to perturbations (Eckhaus stable) the validity has been looked at in \cite{melbourne2004phase}, while near the Eckhaus instability the validity of a Korteweg de Vries equation was proved in \cite{haas2018modulation}. Near the so-called Benjamin-Feir-Newell instability, where all traveling wave solutions are unstable, a fourth-order diffusive equation called the (1D) Kuramoto-Sivashinsky (KS) equation \cite{kuramoto1976persistent,sivashinsky1977derivation} has been established \cite{van2004phase}.

A derivation and partial analysis of the 2D, isotropic extension of the 1D Kuramoto-Sivashinsky (KS) equation as the equation for the phase dynamics of the (2D) CGLE near the Benjamin-Feir-Newell instability is given in \cite{kuramoto1976persistent}. The form of this equation is
\begin{equation}
\partial_T\Phi=-(\partial_X^2\Phi+\partial_Y^2\Phi)-(\partial_X^4\Phi+2\partial_X^2\partial_Y^2\Phi+ \partial_Y^4\Phi)+\frac{1}{2}\big((\partial_X\Phi)^2+(\partial_Y\Phi)^2\big),\label{eq:KSeq_Intro}
\end{equation}
where $(X,Y,T)$ are slow variables. This equation has been studied in its own right as a model for phenomena such as flame fronts  and the leading edge of a viscous fluid flowing down an inclined plane as well as for its spatiotemporal chaotic behavior \cite{Kuramoto1984,kuramoto1976persistent,sivashinsky1977derivation}. If all $Y$-dependent terms are absent, it reduces to the 1D KS equation. 

In this paper we establish anisotropic versions of the KS equation as well as another 2D extension of the 1D KS equation as equations governing the phase dynamics of the ACGLE near different Benjamin-Feir-Newell-type instabilities. The validity of these equations is confirmed through numerical simulations. Beyond its relevance for the ACGLE, an anisotropic 2D KS equation has been been introduced as a model for surface sputter erosion and epitaxial growth \cite{rost1995anisotropic}. For an in-depth numerical study of (\ref{eq:KSeq_Intro}), we refer to \cite{kalogirou2015}.

\section{Linear Stability Analysis}
\label{sec:LSA}

\subsection{Traveling Plane-Wave Solutions and their Perturbations}
\label{subsec:TPWS}

In this section we study the linear stability of traveling plane-wave solutions (TPWS's) of the ACGLE (\ref{eq:ACGLE}). Note that the anisotropy presents itself in the linear dispersion terms.  A TPWS to the ACGLE in the \(\bm{k}\)-direction and with frequency \(\omega\) is given by 
\begin{subequations}\label{eq:ACGLE_travelingwave_solution} 
\begin{align} 
  A &= R_0 e^{i(\bm{k}\cdot\bm{x} - \omega t)}, \label{eq:ACGLE_travelingwave_A}\\
  R_0^2 &= \mu - (k_1^2 + k_2^2), \label{eq:ACGLE_travelingwave_R0}\\
  \omega &= \beta R_0^2 + \alpha_1k_1^2 + \alpha_2k_2^2,\label{eq:ACGLE_travelingwave_omega}
\end{align}
\end{subequations}
where \(\bm{k} = (k_1,k_2)\) is the wavenumber and \(\bm{x} = (x,y)\). The requirements (\ref{eq:ACGLE_travelingwave_R0}) and (\ref{eq:ACGLE_travelingwave_omega}) are found by substituting (\ref{eq:ACGLE_travelingwave_A}) into  (\ref{eq:ACGLE}).

To determine the stability of a TPWS (\ref{eq:ACGLE_travelingwave_solution}) we analyze the time-evolution of perturbations. Separating the perturbations $r(x,y,t)$ of the amplitude and $\phi(x,y,t)$ of the phase as
\begin{equation}\label{eq:ACGLE_travelingwave_perturbed}
  A(x,y,t) = R_0(1+r(x,y,t))e^{i[\bm{k}\cdot\bm{x}- \omega t + \phi(x,y,t)]},  
\end{equation}
substituting (\ref{eq:ACGLE_travelingwave_perturbed}) into (\ref{eq:ACGLE}), and then separating real and imaginary parts leads to the evolution equations
\begin{equation}\label{eq:expanded_amplitude}
  \begin{split}
    \partt{r} &= -2R_0^2r - 3R_0^2r^2 -R_0^2r^3 + \partxx{r} + \partyy{r} - (\partx{\phi})^2 - (\party{\phi})^2 - r(\partx{\phi})^2 - r(\party{\phi})^2 \\
    &\quad -2k_1\alpha_1\partx{r} - 2k_2\alpha_2\party{r} - 2k_1\partx{\phi} - 2k_2\party{\phi} - 2k_1r\partx{\phi} - 2k_2r\party{\phi} \\
    &\quad - \alpha_1\partxx{\phi} - r\alpha_1\partxx{\phi} - \alpha_2\partyy{\phi} - r\alpha_2\partyy{\phi} - 2\alpha_1\partx{r}\partx{\phi} - 2\alpha_2\party{r}\party{\phi},
  \end{split}
\end{equation}
and
\begin{equation}\label{eq:expanded_phase}
  \begin{split}
    \partt{\phi} &= -2R_0^2\beta r - R_0^2\beta r^2 - \alpha_1(\partx{\phi})^2 - \alpha_2(\party{\phi})^2 + \partxx{\phi} + \partyy{\phi} - 2k_1\alpha_1\partx{\phi} - 2k_2\alpha_2\party{\phi}\\
    &\quad + \frac{1}{1+r}\bigg[\alpha_1\partxx{r} + \alpha_2\partyy{r} + 2k_1\partx{r} + 2k_2\party{r} + 2\partx{r}\partx{\phi} + 2\party{r}\party{\phi}\bigg].
  \end{split}
\end{equation}
for the perturbations.
For details of the computation of (\ref{eq:expanded_amplitude}) and (\ref{eq:expanded_phase}) we refer to \cite{Handwerk2019}. 
 
Upon linearizing around the base state  \((r,\phi) = (0,0)\) of the unperturbed TPWS, (\ref{eq:expanded_amplitude}) and (\ref{eq:expanded_phase}) simplify to
\begin{equation}
  \partt{r} = -2R_0^2r + \partxx{r} + \partyy{r} - 2k_1\alpha_1\partx{r} - 2k_2\alpha_2\party{r} - 2k_1\partx{\phi} - 2k_2\partY{\phi} - \alpha_1\partXX{\phi} - \alpha_2\partYY{\phi},\label{eq:lineq_amplitude}
\end{equation}
\begin{equation}
  \partt{\phi} = -2R_0^2\beta r + \partxx{\phi} + \partyy{\phi} - 2k_1\alpha_1\partx{\phi} - 2k_2\alpha_2\party{\phi} + \alpha_1\partxx{r} + \alpha_2\partyy{r} + 2k_1\partx{r} + 2k_2\party{r}.\label{eq:lineq_phase}
\end{equation}
Substituting the modes \(r = \hat{r}e^{\sigma t + i\bm{q}\cdot \bm{x}}\) and \(\phi = \hat{\phi}e^{\sigma t + i\bm{q}\cdot \bm{x}}\) with wavevector $\bm{q}=(q_1,q_2)$ into (\ref{eq:lineq_amplitude}) and (\ref{eq:lineq_phase}) leads to an eigenvalue problem for $(\hat{r},\hat{\phi})$ with eigenvalue $\sigma$. The trace and determinant of the corresponding matrix ${\mathcal M}$ are
\begin{eqnarray}
\mbox{Tr} \; {\mathcal M}&=&-2(R_0^2+|\bm{q}|^2)-4i(\alpha_1k_1q_1+\alpha_2k_2q_2),\label{eq:TrM}\\
\mbox{det} \; {\mathcal M}&=&2q_1^2\big(R_0^2(1+\alpha_1\beta)-2(1+\alpha_1^2)k_1^2\big) +
2q_2^2\big(R_0^2(1+\alpha_2\beta)-2(1+\alpha_2^2)k_2^2\big)\nonumber\\
&&-\:8(1+\alpha_1\alpha_2)k_1k_2q_1q_2+|\bm{q}|^4+(\alpha_1q_1^2+\alpha_2q_2^2)^2\label{eq:detM}\\
&&+\:4iR_0^2\big((\alpha_1-\beta)k_1q_1+(\alpha_2-\beta)k_2q_2\big)+4i(\alpha_1-\alpha_2)q_1q_2(k_1q_2-k_2q_1).\nonumber
\end{eqnarray}

\subsection{Long-Wave Stability and Instability}
\label{subsec:LWS-I}

We are primarily interested in the stability of TPWS's against long-wavelength (LW) perturbations where $|\bm{q}|$ is arbitrarily small, $0<|\bm{q}|\ll 1$. Expanding the two roots of the characteristic equation $\sigma^2-\sigma\mbox{Tr}{\mathcal M}+\mbox{det}{\mathcal M}=0$ for small $(q_1,q_2)$ yields one eigenvalue $\sigma_s=-2R_0^2+{\mathcal O}(|\bm{q}|)$ with negative real part for sufficiently small $|\bm{q}|$, and another eigenvalue that determines the stability of the TPWS against LW perturbations,
\begin{equation}
\begin{split}\label{eval:2}
  \sigma(\bm{q};\bm{k}) &= 2ik_1(\beta-\alpha_1)q_1 + 2ik_2(\beta-\alpha_2)q_2 + \Big(\frac{2k_1^2}{R_0^2} + \frac{2\beta^2k_1^2}{R_0^2} - \alpha_1\beta - 1\Big)q_1^2 \\ &\quad + \Big(\frac{4k_1k_2}{R_0^2} + \frac{4\beta^2k_1k_2}{R_0^2}\Big)q_1q_2 + \Big(\frac{2k_2^2}{R_0^2} + \frac{2\beta^2k_2^2}{R_0^2} - \alpha_2\beta - 1\Big)q_2^2 + {\mathcal O}(|\bm{q}|^3).
  \end{split}
\end{equation}
The real part of this eigenvalue can be written as
\begin{equation}\label{eq:realpart}
\sigma_r(\bm{q};\bm{k})={\mathcal Q}(\bm{q};\bm{k})+{\mathcal O}(|\bm{q}|^4),
\end{equation}
where ${\mathcal Q}(\bm{q};\bm{k})$ is the following quadratic form with respect to $\bm{q}$, with $\bm{k}$ considered as a parameter;
\begin{equation}\label{eq:Q}
{\mathcal Q}(\bm{q};\bm{k})=D_{xx}(\bm{k})q_1^2 + 2D_{xy}(\bm{k})q_1q_2 + D_{yy}(\bm{k})q_2^2,
\end{equation}
with
\begin{equation}\label{eq:ACGLE_instability_condition_coeffs}
  D_{xx}=\frac{2k_1^2(1+\beta^2)}{R_0^2} - (\alpha_1\beta + 1),\;\;\;\;  %(\bm{k})
  D_{xy}=\frac{2k_1k_2(1+\beta^2)}{R_0^2}, \;\;\;\;
  D_{yy}=\frac{2k_2^2(1+\beta^2)}{R_0^2} - (\alpha_2\beta + 1). 
\end{equation}
%
%Consider then a TPWS with wavenumber $\bm{k}=(k_1,k_2)$ (such that $R_0^2>0$). If the symmetric matrix
These coefficients are in agreement with the linear stability analysis of the isotropic CGLE \cite{lega2001traveling}.
Note that for fixed $\bm{k}$, ${\mathcal Q}(\bm{q};\bm{k})$ is the quadratic form with respect to $\bm{q}$ that is associated with the symmetric matrix
\begin{equation}\label{eq:D-matrix}
{\mathcal D}(\bm{k})= \left(\begin{array}{ll}D_{xx}(\bm{k})&D_{xy}(\bm{k})\\D_{xy}(\bm{k})&D_{yy}(\bm{k})\end{array}\right).
\end{equation}

We introduce the following notions of LW-stability and -instability.
\begin{defn}
\label{def1}
A TPWS with wavenumber $\bm{k}$ is 
\begin{compactitem}
\item[(i)] LW-stable if $\sigma_r(\bm{q};\bm{k})<0$ for all sufficiently small, nonzero wavenumbers $\bm{q}$,
\item[(ii)] partly LW-unstable if there exist arbitrarily small, nonzero wavenumbers $\bm{q},\tilde{\bm{q}}$ such that $\sigma_r(\bm{q};\bm{k})<0$ and $\sigma_r(\tilde{\bm{q}};\bm{k})>0$,%\\[-20pt]
\item[(iii)] fully LW-unstable if $\sigma_r(\bm{q};\bm{k})>0$ for all sufficiently small, nonzero wavenumbers $\bm{q}$,
\item[(iv)] LW-unstable if (ii) or (iii) hold.
\end{compactitem}
\end{defn}
\noindent 
Since $\sigma_r(\bm{0};\bm{k})=0$, property (i) is satisfied if $\bm{q}=(0,0)$ is a strict local maximum of $\sigma_r(\bm{q};\bm{k})$ for fixed $\bm{k}$, which is the case if ${\mathcal D}(\bm{k})$ is negative definite. Similarly, property (iii) is satisfied if $\bm{q}=(0,0)$ is a strict local minimum of $\sigma_r(\bm{q};\bm{k})$ for fixed $\bm{k}$, which is the case ${\mathcal D}(\bm{k})$ is negative definite.  ${\mathcal D}(\bm{k})$ is negative (positive) definite if and only if its determinant is positive and one of the diagonal entries $D_{xx}$, $D_{yy}$ is negative (positive). A sufficient condition for (ii) is that the determinant of ${\mathcal D}(\bm{k})$ be negative, in which which case ${\mathcal Q}(\bm{q};\bm{k})$ defines a saddle surface so that there are regions in the $(q_1,q_2)$-plane in which ${\mathcal Q}(\bm{q};\bm{k})>0$ and ${\mathcal Q}(\bm{q};\bm{k})<0$. Thus, to analyze the stability properties of a given TPWS, we have to take the determinant
\begin{equation}\label{eq:detD}
\mbox{det} \; {\mathcal D}(\bm{k})=(\alpha_1\beta+1)(\alpha_2\beta+1)-\frac{2(1+\beta^2)}{R_0^2}\Big((\alpha_1\beta+1)k_2^2+(\alpha_2\beta+1)k_1^2)\Big)
\end{equation}
into consideration. 

The stability properties (i)-(iii) occur in the following parameter regimes:
\begin{theorem}
\label{thm1}
Suppose that $(1+\alpha_1\beta)(1+\alpha_2\beta)\neq 0$ and let $\bm{k}$ be the wavenumber of a TPWS such that $R_0^2>0$. Define $F(\bm{k})$ by
\begin{equation}\label{eq:F}
F(\bm{k})=\Bigg(1+\frac{2(1+\beta^2)}{1+\alpha_1\beta}\Bigg)k_1^2+\Bigg(1+\frac{2(1+\beta^2)}{1+\alpha_2\beta}\Bigg)k_2^2,
\end{equation} 
and assume that the coefficients of $k_1^2$ and $k_2^2$ are both nonzero. 
\begin{compactitem}
\item[(a)] Suppose $1+\alpha_1\beta>0$ and $1+\alpha_2\beta>0$. Then $F(\bm{k})=\mu$ defines an ellipse that is inscribed in the $\mu$-circle $|\bm{k}|^2=\mu$. If $\bm{k}$ is in the interior of that ellipse the TPWS is LW-stable while for $\bm{k}$ outside of the ellipse the TPWS is partly LW-unstable. 
\item[(b)] Suppose $1+\alpha_1\beta<0$ and $1+\alpha_2\beta>0$. Then the curve $F(\bm{k})=\mu$ is either a hyperbola or an ellipse. In either case this curve intersects the $\mu$-circle in the four points defined by
\begin{equation}\label{eq:intersectionpoint}
k_1^2=\frac{-(1+\alpha_1\beta)\mu}{(\alpha_2-\alpha_1)\beta},\;\;\;\;k_2^2=\frac{(1+\alpha_2\beta)\mu }{(\alpha_2-\alpha_1)\beta}.
\end{equation}
Moreover, the TPWS is partly unstable if $F(\bm{k})<\mu$ and fully unstable if $F(\bm{k})>\mu$. 
\item[(c)] Suppose $1+\alpha_1\beta<0$ and $1+\alpha_2\beta<0$. Then every TPWS is fully unstable. 
\end{compactitem}
\end{theorem}

\begin{proof}
To simplify notation we set  $\tilde{\alpha}_j=1+\alpha_j\beta$, $j=1,2$. When multiplying the expression (\ref{eq:detD}) for $\mbox{det} \; {\mathcal D}(\bm{k})$ by $R_0^2/(\tilde{\alpha}_1\tilde{\alpha}_2)$ and substituting $R_0^2=\mu-k_1^2-k_2^2$ one can see directly that

\begin{compactitem}
\item[(i)] if $\tilde{\alpha}_1\tilde{\alpha}_2>0$, then $\mbox{det} \; {\mathcal D}(\bm{k})>0$ if and only if $\mu>F(\bm{k})$, and $\mbox{det} \; {\mathcal D}(\bm{k})<0$ if and only if $\mu<F(\bm{k})$;
\item[(ii)] if $\tilde{\alpha}_1\tilde{\alpha}_2<0$, then $\mbox{det} \; {\mathcal D}(\bm{k})>0$ if and only if $\mu<F(\bm{k})$, and $\mbox{det} \; {\mathcal D}(\bm{k})<0$ if and only if $\mu>F(\bm{k})$. 
\end{compactitem}

%\begin{itemize}

%\item[(a)] 
Part (a) then follows directly from (i) as for $\tilde{\alpha}_1,\tilde{\alpha}_2$ both positive we have $\mbox{det} \; {\mathcal D}(\bm{k})>0$ inside the ellipse $\mu=F(\bm{k})$, which makes ${\mathcal D}(\bm{k})$ negative definite, and $\mbox{det} \; {\mathcal D}(\bm{k})<0$ outside of this ellipse. That the ellipse is inside the $\mu$-circle is clear as the two coefficients of $F(\bm{k})$ are both $>1$. 

%\item[(c)] 
Part (c) also follows from (i) since for $\tilde{\alpha}_1,\tilde{\alpha}_2$ both negative we have $\mu>k_1^2+k_2^2>F(\bm{k})$, hence $\mbox{det} \; {\mathcal D}(\bm{k})$ is always positive. Since $\tilde{\alpha}_1<0$ implies $D_{xx}(\bm{k})>0$, ${\mathcal D}(\bm{k})$ is positive definite if $\tilde{\alpha}_1,\tilde{\alpha}_2$ are both negative.

%\item[(b)] 
Lastly, part (b) follows from (ii) in the same way as (a) follows from (i). Computing the intersections of the curve $\mu=F(\bm{k})$ with the circle $\mu=|\bm{k}|^2$ is straightforward.

%\end{itemize}
\end{proof}

The case distinctions in Definition \ref{def1} and Theorem \ref{thm1} are formulated with strict inequalities and so do not include borderline cases. We treat borderline cases such as $1+\alpha_1\beta=0$ or $F(\bm{k})=\mu$ as boundary sets of positive codimension in parameter space. Case (b) of Theorem \ref{thm1} actually covers two cases, with the second case obtained via the parameter swap $(\alpha_1,k_1,q_1)\leftrightarrow (\alpha_2,k_2,q_2)$. 

In Figure \ref{fig:caseb}a-d we illustrate case (b) of Theorem \ref{thm1}. The division of the $(k_1,k_2)$-plane is depicted in a and b for the cases when $F(\bm{k})=\mu$ defines hyperbolic and elliptic curves, respectively. In Figure \ref{fig:caseb}c,~d we show the division of the $(q_1,q_2)$-plane into regions with ${\mathcal Q}>0$ and ${\mathcal Q}<0$ for the two points $\bm{k}=(0,0)$ and $\bm{k}=(k_1^\ast,k_2^\ast)$, which is close to the $F(\bm{k})=\mu$ curve, marked in Figure \ref{fig:caseb}a. When $\bm{k}$ approaches this curve the two lines separating the four regions merge and the region ${\mathcal Q}<0$ disappears as $\mbox{det} \; {\mathcal D}(\bm{k})\rightarrow 0$.

\subsection{Stability Boundaries and Effect of Anisotropy on Chaotic Solutions}
\label{subsec:boundaries}

\paragraph{Eckhaus stability boundary}
The ellipse $F(\bm{k})=\mu$ in case (a) of Theorem \ref{thm1} extends the Exkhaus stability boundary for the CGLE in the 1D and 2D cases (see, e.g. \cite{Aranson2002,DaKra1998}) to the 2D anisotropic case of the ACGLE. In particular, in the 2D isotropic case ($\alpha_1=\alpha_2$) the ellipse becomes a circle. The elliptic stability boundary along with the conditions $1+\alpha_j\beta>0$, $j=1,2$, for the ACGLE was already established in \cite{dangelmayr2008modulational} in a more general setting and using a different notation. In \cite{dangelmayr2008modulational}, in addition to the LW-stability boundary, short-wavelength instabilities have been analyzed that may preceed the LW-instability when $|\bm{k}|$ is increased along a ray emanating from the origin, thereby extending the stability analysis pursued in \cite{VanHarten1995} for the 1D case. 

\paragraph{Benjamin-Feir-Newell stability boundary}
The stability analysis in \cite{dangelmayr2008modulational} is exclusively for the case when stable TPWS's exist, that is, when $1+\alpha_1\beta>0$ and $1+\alpha_2\beta>0$. These two conditions extend the Newell criterion \cite{Newell1974}, $1+\alpha\beta>0$, for the existence of stable TPWS's of the CGLE, when $\alpha_1=\alpha_2\equiv\alpha$, to the anisotropic case. For the CGLE, the curve in the $(\alpha,\beta)$-plane defined by $1+\alpha\beta=0$ is referred to as the Benjamin-Feir-Newell- (BFN-)stability boundary \cite{Aranson2002,Chate1996}. A comprehensive numerical study of the CGLE in the BFN-unstable as well as BFN-stable-regimes was performed in \cite{Chate1996}.

In the anisotropic case of the ACGLE (\ref{eq:ACGLE}), the boundary separating the parameter region with no stable TPWS's  from the region in which stable TPWS's exist is given by the condition $1+\alpha_1\beta=0$ or $1+\alpha_2\beta=0$. This defines two surfaces in the $(\alpha_1,\alpha_2,\beta)$-space that intersect in the plane $\alpha_1=\alpha_2$ in the BFN-stability boundary for the CGLE. These two surfaces, which we also refer to as the BFN-stability boundary, along with the regions from Theorem \ref{thm1}, are visualized in Figure \ref{fig:BFNsurface}. In Section \ref{sec:WNA}  we will study phase equations governing the evolution of solutions $A(x,y,t)$ to the  ACGLE for parameters near near the BFN-stability boundary.
%We will assume that $1+\alpha_1\beta\le 0$ and consider $1+\alpha_2\beta > 0$ as well as $1+\alpha_2\beta\le 0$.
%In this paper we are primarily interested in the case when no stable TPWS's exist for the ACGLE. 

\begin{figure}[t]%[!htbp] % don't put [p] here!!=
\begin{center}
\mbox{}\hspace{-20pt}
  \subfloat[]{\includegraphics[width=.43\textwidth]{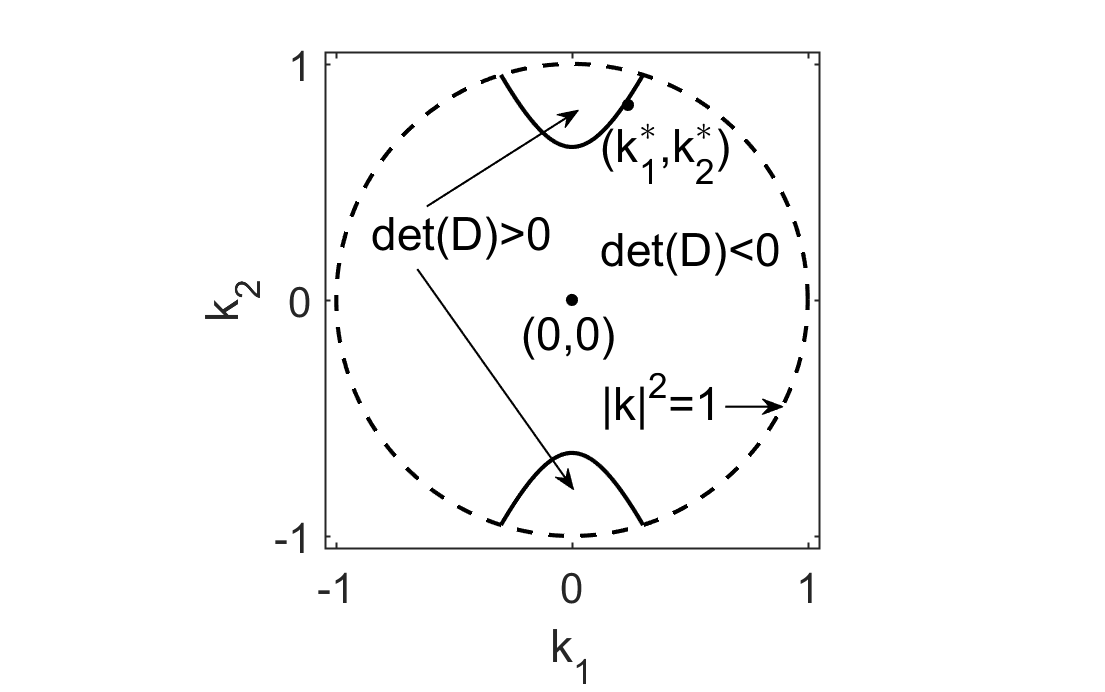}} \hspace{-20pt} 
  \subfloat[]{\includegraphics[width=.36\textwidth]{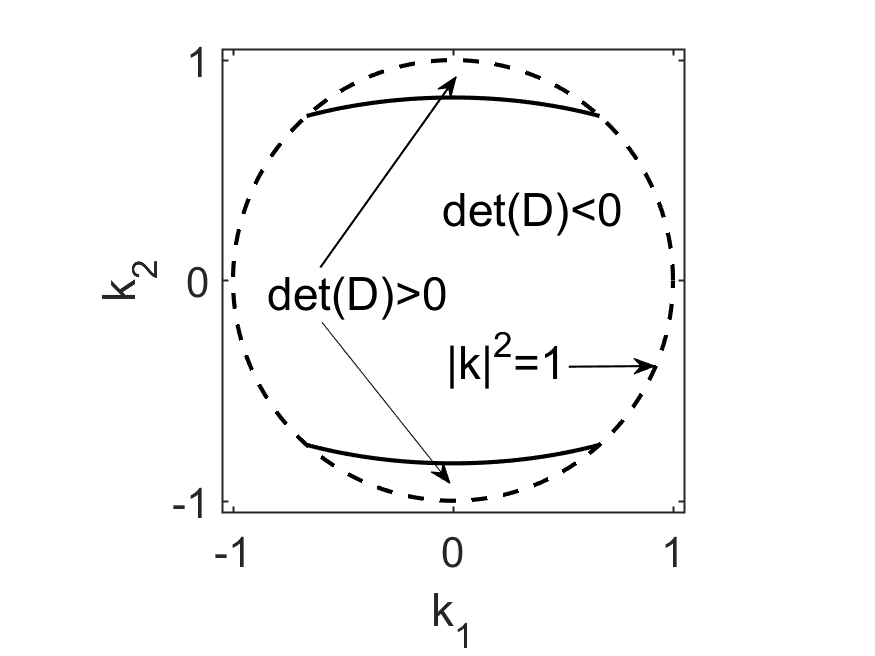}}\\
	\subfloat[]{\includegraphics[width=.36\textwidth]{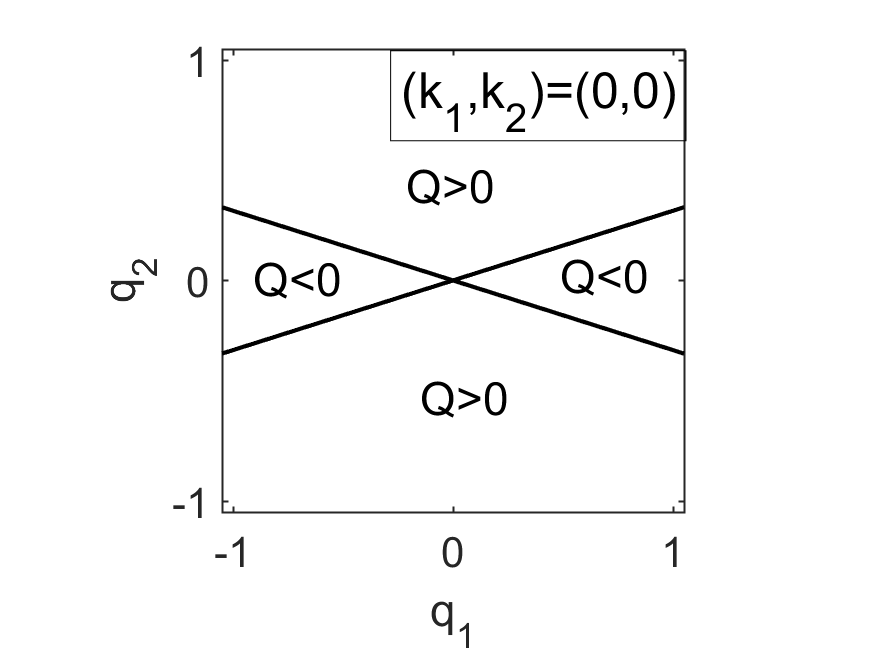}} 
  \subfloat[]{\includegraphics[width=.36\textwidth]{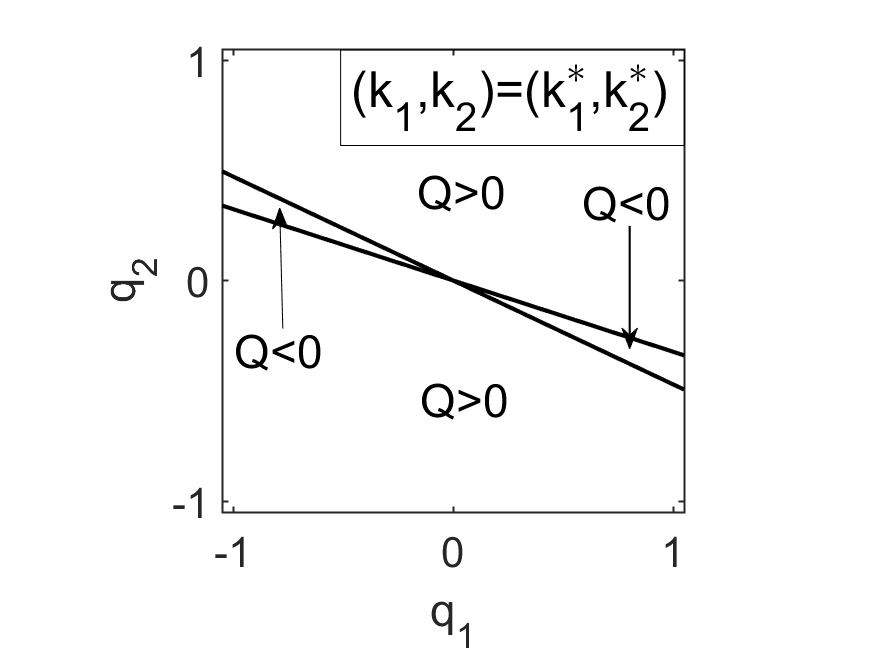}}
	\end{center}
  \caption[]{(a) Circle $|\bm{k}|=\mu=1$  (dashed) and the segments of the curve $F(\bm{k})=1$ (solid) that are inside that circle for $\alpha_1=-1.2$, $\alpha_2=2$, $\beta=1.1$. For these parameters $F(\bm{k})=\mu$ defines hyperbolae. The regions in the $(q_1,q_2)$-plane where ${\mathcal Q}>0$ and ${\mathcal Q}<0$ are depicted in (c) and (d) for the points $(k_1,k_2)=(0,0)$ and $(k_1^\ast,k_2^\ast)=(0.2375,0.8261)$ marked in (a). (b) Same as (a) for $\alpha_2=-\alpha_1=8$, $\beta=1.1$; here $F(\bm{k})=\mu$ defines an ellipse. 
	} 
  \label{fig:caseb}
\end{figure}

\begin{figure}[htbp]%[!htbp]
	\centering
  \includegraphics[width=.5\textwidth]{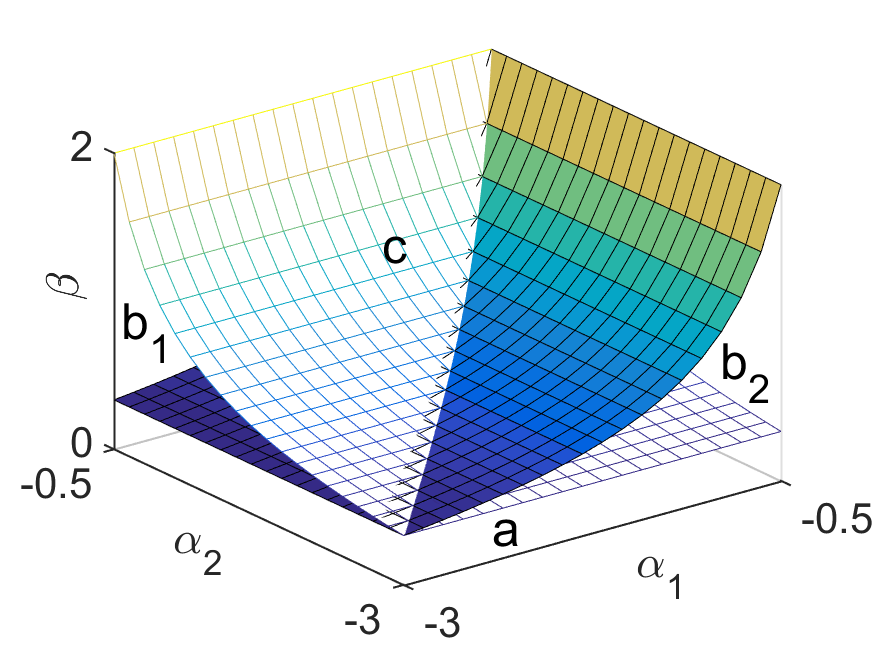}
  \caption[]{Separation surfaces $1+\alpha_1\beta=0$, $1+\alpha_2\beta=0$ for $\beta>0$ and regions where the cases (a), (b), (c) from Theorem \ref{thm1} occur. For (b) the subscripts (b$_1$) and (b$_2$) distinguish $1+\alpha_1\beta<0$ and $1+\alpha_2\beta<0$, respectively.}
	\label{fig:BFNsurface}
\end{figure}

\paragraph{Linear and nonlinear phase equations}
The eigenvalue (\ref{eval:2}) governs the linearization of the evolution equation for the spatial Fourier transform $\hat{\phi}$, with wavenumbers $\bm{q}=(q_1,q_2)$, of the phase $\phi$ with $\partial_t\hat{\phi}=\sigma(\bm{q};\bm{k})\hat{\phi}$. Using the long-wave base-state \(k_1=k_2=0\) (bulk oscillation, for this state $\sigma$ has no odd powers of $\bm{q}$), truncating the expansion of $\sigma$ at fourth order, and then taking the inverse Fourier transform of the truncated equation for $\hat{\phi}$ yields the linear phase equation
\begin{equation}\label{eq:lieanr_phase_aCGL}
  \partt{\phi} = (1+\alpha_1\beta)\partxx{\phi} + (1+\alpha_2\beta)\partyy{\phi} - \frac{\alpha_1^2(1+\beta^2)}{2R_0^2}\partxxxx{\phi} - \frac{\alpha_2^2(1+\beta^2)}{2R_0^2}\partyyyy{\phi} - \frac{\alpha_1\alpha_2(1+\beta^2)}{R_0^2}\partxxyy{\phi}.
\end{equation}
This equation can be extended to a nonlinear equation for $\phi$ using the facts that (since \(k_1=k_2=0\)) the equation must be invariant under $x\rightarrow -x$ and $y\rightarrow -y$, and that it depends only on spatial derivatives. The lowest-order nonlinear terms that satisfy these conditions are $(\partial_x\phi)^2$ and $(\partial_y\phi )^2$. Adding these terms to (\ref{eq:lieanr_phase_aCGL}), we arrive at
\begin{equation}\label{eq:nonlinear_phase_aCGL}
  \begin{split}
  \partt{\phi} &= (1+\alpha_1\beta)\partxx{\phi} + (1+\alpha_2\beta)\partyy{\phi} - \frac{\alpha_1^2(1+\beta^2)}{2R_0^2}\partxxxx{\phi} - \frac{\alpha_2^2(1+\beta^2)}{2R_0^2}\partyyyy{\phi} \\ &\quad - \frac{\alpha_1\alpha_2(1+\beta^2)}{R_0^2}\partxxyy{\phi} + g_0(\partx{\phi})^2 + h_0(\party{\phi})^2,
  \end{split}
\end{equation}
with yet unknown coefficients \(g_0\) and \(h_0\). From the phase equation (\ref{eq:nonlinear_phase_aCGL}), we see that the BFN-instability \(1 + \alpha\beta<0\) for the CGLE manifests itself now in each of the diffusion terms with different \(\alpha_i\) values because of the anisotropy in the ACGLE (\ref{eq:ACGLE}). This again demonstrates the possibility of a traveling plane wave  to be stable in one direction and unstable in the other giving rise to chaotic solutions that do not occur for the isotropic CGLE. 

In Section \ref{sec:WNA} we derive a phase equation of the form of (\ref{eq:nonlinear_phase_aCGL}) for slow variables using a multiple-scale expansion. %{\bf Not sure how hseful this para is.}

\paragraph{Examples for the effect of anisotropy on chaotic solutions}
As pointed out previously, the solutions of the ACGLE encompass those of the  CGLE (for \(\alpha_1 = \alpha_2\)), but the added anisotropy allows for solutions that are not possible in the isotropic case. Not only can traveling plane waves have different degrees of stability or instability in the \(x\) and \(y\) directions, they can now be both stable and unstable to perturbations depending on direction, as captured by Theorem \ref{thm1}(b) and illustrated in Figure \ref{fig:caseb}. For example, in the CGLE phase chaos demonstrates itself as an evolving cellular structure \cite{Chate1996}. This type of phase chaos can exist in the ACGLE, but as one of the linear dispersion parameters, say $\alpha_2$, is adjusted so that, as traveling waves become BFN stable in one direction, the isotropic cellular behavior gives way to a phase-chaotic structure of ripples which are aligned along the stable direction; see Figure \ref{fig:ACGL_chaos}a (all numerical solutions shown in this paper are for $\mu=1$). A similar behavior also happens for parameter values which yield defect chaos in the CGLE. As one direction is made BFN-stable, ripples appear aligned along the stable direction. Defects appear and travel mostly along the unstable direction, see Figure \ref{fig:ACGL_chaos}b. 

Figure \ref{fig:hole_solutions} illustrates the effect of different linear dispersion coefficients on hole and shock wall solutions of the isotropic CGLE.

\begin{figure}[t]%[!htbp] %[t]
	\centering
  	\includegraphics[width=.49\textwidth]{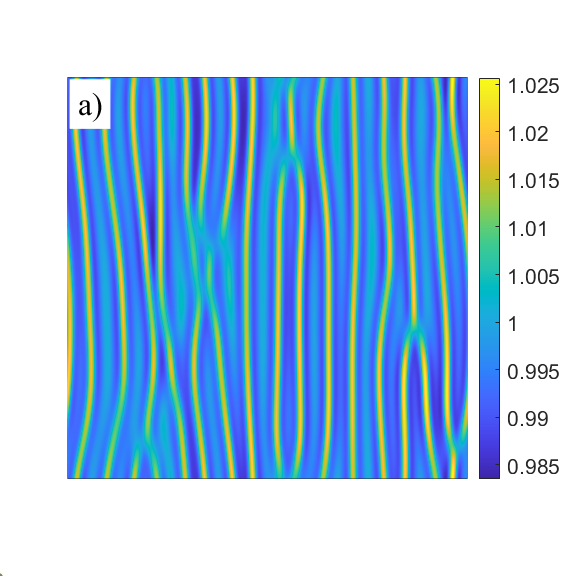} \hfill
  	\includegraphics[width=.49\textwidth]{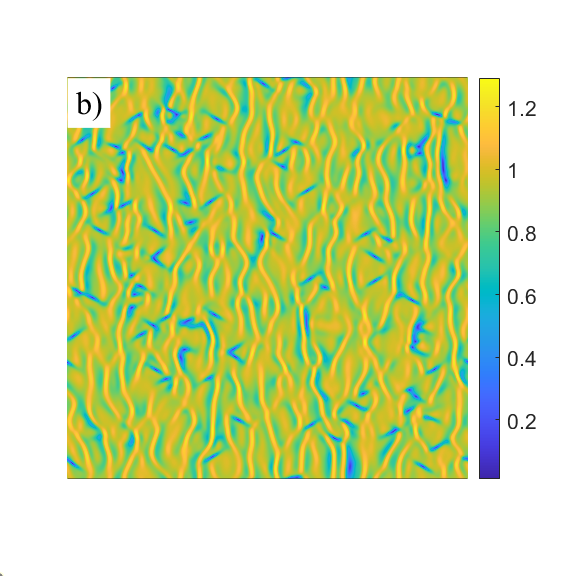}
  	\caption[Phase and defect solutions to the ACGLE: $\mbox{Re}A$???]{Snapshots of chaotic solutions to the ACGLE, where traveling plane waves are stable in the \(y\) direction and unstable in the \(x\) direction. Both figures show \(|A|\). (a) Time snapshot of a simulation with parameter values $\alpha_1 = -1.2,$ $\alpha_2 = 2$, $\beta = 1.1$, at time $t= 600$.  This simulation shows phase chaos where the amplitude is bounded away from zero, \emph{i.e.} \(|A|>0\), and is seen to be around \(1\). The cellular structure of the phase chaotic solution of the CGLE, see \cite{Chate1996}, has been replaced with a ripple structure. (b) Time snapshot of a simulation with parameter values $\alpha_1 = -3,$ $\alpha_2 = 2$, $\beta = 1.1$, at time $t= 600$. The stability in the \(y\) direction is also apparent in this defect-chaos parameter regime, which shows partial coherence in the stable direction. The regions of defects appear to travel along the unstable horizontal direction.}
  \label{fig:ACGL_chaos}
\end{figure}
\begin{figure}[t] 
  \includegraphics[width=.24\textwidth]{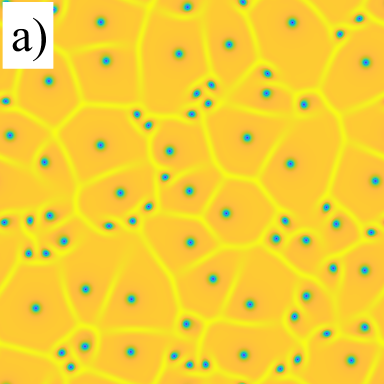} \hfill
  \includegraphics[width=.24\textwidth]{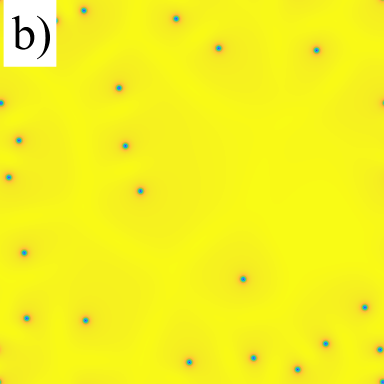} \hfill
  \includegraphics[width=.24\textwidth]{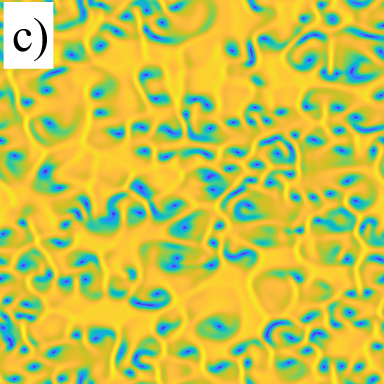} \hfill
  \includegraphics[width=.24\textwidth]{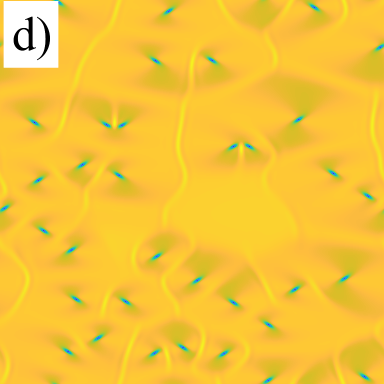}
  \caption[Sample of ACGLE solutions]{Four snapshots of solutions to the ACGLE (\ref{eq:ACGLE}). a) Hole and shock-wall solution. The shock walls keep the spiral defects separated. Here \(\alpha_1 = \alpha_2 = -0.22\). b) The shock walls no longer contain the spirals which are free to diffuse and annihilate. Here \(\alpha_1 = \alpha_2 = 0.22\). c) Spiral-defect chaos with \(\alpha_1 = -5\cdot 0.22\) and \(\alpha_2 = -0.22\). d) The previous 3 solution types can occur in the isotropic case, but this is a uniquely anisotropic solution. There are spiral defects, whose centers are skewed into ellipses by the anisotropy, together with phase-chaotic ripples in the \(y\)-direction. All solutions have \(\beta = 1.1\) and were solved to \(t = 600\) on a grid of \([-100,100]\times[-100,100]\) except d) which is seen on a grid of \([-150,150]\times[-150,150]\) to better capture the behavior. These hole-type solutions empirically require a different initial condition compared to the solutions in Figures \ref{fig:ACGL_chaos}, \ref{fig:ACGLE_AvsPhi_Isotropic}, \ref{fig:ACGLE_AvsPhi_Anisotropic_codim1}, and \ref{fig:ACGLE_AvsPhi_Anisotropic_codim2} below.}
  \label{fig:hole_solutions}
\end{figure}

\section{Weakly Nonlinear Analysis for Long-Wave Instability}
\label{sec:WNA}

In this section, we first establish a general phase equation for the ACGLE (\ref{eq:ACGLE}) by applying a multiple-scale expansion to the perturbational amplitude $r$ and phase $\phi$ in (\ref{eq:ACGLE_travelingwave_perturbed}) for $k_1=k_2=0$. The resulting system of equations allows the amplitude to be slaved to the phase, leaving an equation for the phase alone. In this derivation, no assumptions are yet made about the parameters. We then specify this phase equation for the case when $(\alpha_1,\alpha_2,\beta)$ is close to the BFN-stability boundary, for which we distinguish two cases: the case $1+\alpha_1\beta=0$ and $1+\alpha_2\beta>0$ (codimension-one, Subsection \ref{subsec:cod1}) and the case $1+\alpha_1\beta=0$ and $1+\alpha_2\beta=0$ (codimension-two, Subsection \ref{subsec:cod2}).

Substituting \(\bm{k}=(0,0)\) into the traveling plane wave solution (\ref{eq:ACGLE_travelingwave_solution}) leads to the following ansatz for $A$ with perturbed amplitude and phase,
\begin{equation}\label{eq:aCGL_ansatz_longwave}
    A(x,y,t) = \sqrt{\mu}(1+r(x,y,t))e^{i(- \beta\mu t + \phi(x,y,t))}.
\end{equation}
For the slow space and time scalings, we use
\begin{equation}\label{eq:aCGL_phase_scalings_phi}
    r(x,y,t) = \delta^6 W(\delta x, \delta y, \delta^4 t), \qquad \phi(x,y,t) = \delta^2 \Phi(\delta x, \delta y, \delta^4 t),
\end{equation}
extending the scaling introduced in \cite{DSSS2009} for the 1D CGLE to the ACGLE. After separating real and imaginary parts, scaling, and dividing the \(W\) and \(\Phi\) equations each by \(\delta^6\), we obtain
\begin{align}
    \begin{split}
      \delta^4\partT{W} &= \delta^2(\partXX{W} + \partYY{W}) - 2\delta^4(\alpha_1\partX{W}\partX{\Phi} + \alpha_2\partY{W}\partY{\Phi})\\ &\qquad  + (1+\delta^6W)(-2\mu W - \delta^6\mu W - (\partX{\Phi})^2 - (\partY{\Phi})^2 - \delta^{-2}\alpha_1\partXX{\Phi} - \delta^{-2}\alpha_2\partYY{\Phi}),
    \end{split}\label{eq:longwave_W_with_W}\\
    \begin{split}
      \partT{\Phi} &= -2\beta\mu W - \delta^6\beta\mu W^2 - \alpha_1(\partX{\Phi})^2 - \alpha_2(\partY{\Phi})^2 + \delta^{-2}(\partXX{\Phi} + \partYY{\Phi})\\ &\qquad + \frac{\delta^2(\alpha_1\partXX{W} + \alpha_2\partYY{W}) + 2\delta^4(\partX{W}\partX{\Phi} + \partY{W}\partY{\Phi})}{1+\delta^6W}, \label{eq:longwave_Phi_with_W}
    \end{split}
\end{align}
where $(X,Y,T)=(\delta x,\delta y, \delta^4 t)$. The leading-order solution of the \(W\) equation is
\begin{equation}\label{eq:Wapprox_func_of_Phi}
    W = \frac{-(\alpha_1\partXX{\Phi} + \alpha_2\partYY{\Phi})}{2\delta^2\mu} + \mathcal{O}(1).
\end{equation}
Refining this solution up to $\mathcal{O}(1)$ gives
\begin{equation}\label{eq:longwave_W_sub_soln}
    W = \frac{-(\alpha_1\partXXXX{\Phi} + (\alpha_1 + \alpha_2)\partXXYY{\Phi} + \alpha_2\partYYYY{\Phi})}{4\mu^2} - \frac{(\partX{\Phi})^2 + (\partY{\Phi})^2}{2\mu} - \frac{\alpha_1\partXX{\Phi} + \alpha_2\partYY{\Phi}}{2\delta^2\mu} + \mathcal{O}(\delta^2).
\end{equation}
Equation (\ref{eq:longwave_W_sub_soln}) can be used to eliminate \(W\) from the equation (\ref{eq:longwave_Phi_with_W}) for \(\Phi\). Doing so yields
\begin{equation}\label{eq:2D_aCGL_phase_full}
    \begin{split}
      \partT{\Phi} &= \frac{(1+\alpha_1\beta)\partXX{\Phi} + (1+\alpha_2\beta)\partYY{\Phi}}{\delta^2} \\ &\qquad+ \frac{(\beta\alpha_1-\alpha_1^2)\partXXXX{\Phi} + \big(\beta(\alpha_1+\alpha_2)-2\alpha_1\alpha_2\big)\partXXYY{\Phi} + (\beta\alpha_2-\alpha_2^2)\partYYYY{\Phi}}{2\mu}\\ &\qquad + (\beta-\alpha_1)(\partX\Phi)^2 + (\beta-\alpha_2)(\partY\Phi)^2+ \mathcal{O}(\delta^2).
    \end{split}
\end{equation}
%where the fourth-order derivative terms that are independent of \(\beta\) come from the leading-order terms of $\delta^2\alpha_1\partXX{W}/(1+\delta^6W)$ and 
%$\delta^2\alpha_2W/(1+\delta^6W)$. {\bf Omit? or check that this is correct.}%{\bf omit?}

We have now found the coefficients for the nonlinear terms in Equation (\ref{eq:nonlinear_phase_aCGL}). The linear coefficients of (\ref{eq:nonlinear_phase_aCGL}) and (\ref{eq:2D_aCGL_phase_full}) are equal on the BFN neutral stability curve \(1 + \alpha_1\beta = 1 + \alpha_2\beta = 0\) and match the form of the isotropic phase equation given in \cite{lega2004phase}.

As in (\ref{eq:nonlinear_phase_aCGL}), the anisotropy allows for the possibility of a Benjamin-Feir-Newell type instability to occur in the \(X\) or \(Y\) directions individually, or for both directions to become unstable simultaneously. This naturally leads to codimension-one and codimension-two phase equations.

The solutions of the phase equation can be compared to the solutions of the phase of solutions of the ACGLE, as seen in Figure \ref{fig:ACGLE_AvsPhi_Isotropic} which shows isotropic phase chaos. The phase of \(A\) should look similar to the solution of the phase equation \(\Phi\). We can use the approximation (\ref{eq:Wapprox_func_of_Phi}) for \(W(\Phi)\) after solving the phase equation for \(\Phi\) and recreate solutions of the ACGLE by substituting \(W(\Phi)\) and \(\Phi\) into the ansatz (\ref{eq:aCGL_ansatz_longwave}) using (\ref{eq:aCGL_phase_scalings_phi}).

\begin{figure}[!htbp] 
  \includegraphics[width=.24\textwidth]{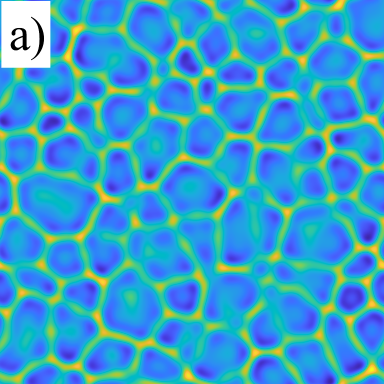} \hfill
  \includegraphics[width=.24\textwidth]{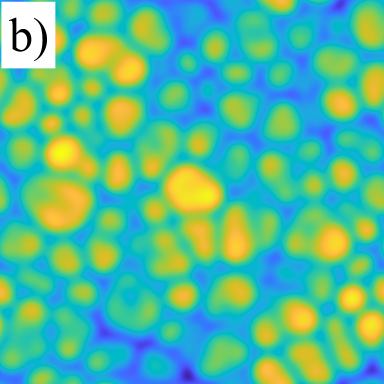} \hfill
  \includegraphics[width=.24\textwidth]{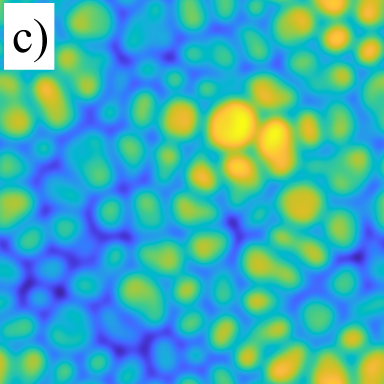} \hfill
  \includegraphics[width=.24\textwidth]{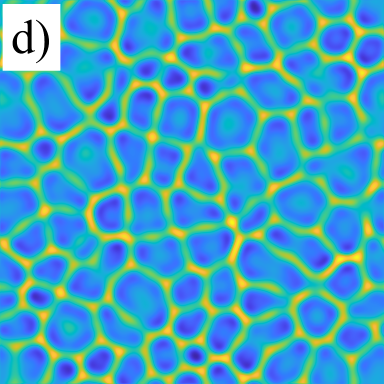}   
  \caption[Comparison of ACGLE and Phase Equation]{A solution $A(x,y,t)$ to the ACGLE with parameters \(\alpha_1=\alpha_2=-0.22\), \(\beta=5\) was simulated on a square domain \([-100,100]\times[-100,100]\) to time \(t=T/\delta^4 \approx 6200\). (a) Absolute value $|A|$ and (b) phase angle $\texttt{angle}(A)$ of $A$. (c) A solution \(\Phi(X,Y,T)\) of the phase equation (\ref{eq:2D_aCGL_phase_full}) with $\delta=0.3$ at time \(T=50\) with the same parameters used to solve \(A\). The \(\Phi\) calculated from the phase equation displays qualitatively the same behavior as the phase of the ACGLE with these isotropic parameters. (d) The absolute value $|A|$ for $A$ constructed from \(\Phi\) using Equations (\ref{eq:aCGL_ansatz_longwave}), (\ref{eq:aCGL_phase_scalings_phi}) and (\ref{eq:Wapprox_func_of_Phi}).}
  \label{fig:ACGLE_AvsPhi_Isotropic}
\end{figure}

\subsection{Codimension-One Bifurcation}
\label{subsec:cod1}

For the codimension-one case, the condition for the \(y\)-direction to be BFN-stable is that \(1+\alpha_2\beta > 0\). Regarding $1+\alpha_1\beta$, we consider the situation in which we are close to the BFN-stability boundary $1+\alpha_1\beta=0$. Extending again the scaling used in \cite{DSSS2009} for the 1D CGLE to our 2D anisotropic case we unfold this degeneracy by setting 
\begin{equation}
1+\alpha_1\beta =\kappa\delta^2,
\end{equation}
where $\kappa$ is treated as ${\mathcal O}(1)$-parameter. Note that for $\kappa<0$ the original parameters are in region (b$_1$) in Figure \ref{fig:BFNsurface}. Since we have BFN stability in the $y$ direction, we have to  balance the $\delta^4$-scaling of $t$ by a $\delta^2$-scaling of $y$. Thus, we set \(\tilde{Y} = \delta Y = \delta^2 y\), and use \(Y = \tilde{Y}/\delta\) in (\ref{eq:2D_aCGL_phase_full}). Omitting the tilde, we arrive at
\begin{equation}\label{eq:ACGLE_codim1_phase_with_a1}
	\partT{\Phi} = \kappa\partXX{\Phi} + (1+\alpha_2\beta)\partYY{\Phi} + \frac{(\beta\alpha_1-\alpha_1^2)\partXXXX{\Phi}}{2\mu} + (\beta-\alpha_1)(\partX\Phi)^2 + \mathcal{O}(\delta^2),
\end{equation}
with \(\kappa\) as a control parameter. Assuming $\beta\ne 0$, using $\alpha_1=-1/\beta+{\mathcal O}(\delta^2)$ in the fourth-order derivative and nonlinear terms, and omitting ${\mathcal O}(\delta^2)$ gives the final codimension-one phase equation
\begin{equation}\label{eq:2D_aCGL_phase_codim1}
	\partT{\Phi} = \kappa\partXX{\Phi} + (1+\alpha_2\beta)\partYY{\Phi} - \frac{1}{2\mu}\bigg(1+\frac{1}{\beta^2}\bigg)\partXXXX{\Phi} + (\beta+\frac{1}{\beta})(\partX\Phi)^2,
\end{equation}
which retains the nonlinear dispersion parameter \(\beta\) and the \(y\)-direction linear dispersion parameter \(\alpha_2\) from the ACGLE (\ref{eq:ACGLE}). We emphasize again that for (\ref{eq:2D_aCGL_phase_codim1}) to be applicable it is required that \(1+\alpha_2\beta\) be positive and ${\mathcal O}(1)$ so that the $Y$-diffusion coefficient stays positive and no fourth-order linear derivative terms or nonlinear terms with respect to $Y$ are needed to saturate the instability.

Assuming $\kappa<0$ and using the rescaling 
$$\hat{T}=-\kappa\hat{\mu}T,\;\;\:\hat{X}=\sqrt{\hat{\mu}}X,\;\;\;\hat{Y}=\sqrt{\frac{-\kappa\hat{\mu}}{1+\alpha_2\beta}}Y,\;\;\;\hat{\Phi}=-\frac{2}{\kappa}\big(\beta+\frac{1}{\beta}\big)\Phi,$$
where $\hat{\mu}=\frac{-2\kappa\mu}{1+1/\beta^2}$, all coefficients in (\ref{eq:2D_aCGL_phase_codim1}) become normalized, and the equation simplifies to (with the hats omitted)
\begin{equation}
\partial_T\Phi=-\partial_X^2\Phi-\partial_X^4\Phi+\frac{1}{2}(\partial_X\Phi)^2+\partial_Y^2\Phi.
\label{eq:2D_aCGL_phase_codim1_rescaled}
\end{equation}
The fact that all coefficients in the codimension-one phase equation can be normalized means that there is a unique phase dynamics (apart from varying initial conditions) that controls the ACGLE-dynamics for generic $1+\alpha_2\beta>0$ and sufficiently small negative values of $1+\alpha_1\beta$. Note that Equation (\ref{eq:2D_aCGL_phase_codim1_rescaled}) is just the standard 1D KS-equation augmented by a diffusion-term in $Y$.

\begin{figure}[!htbp] 
  \includegraphics[width=.24\textwidth]{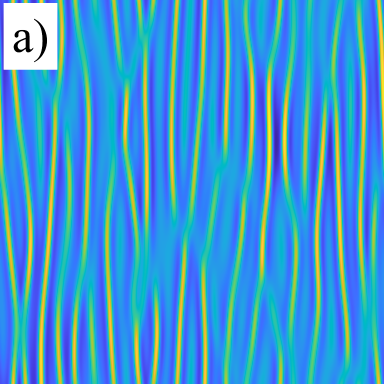} \hfill
  \includegraphics[width=.24\textwidth]{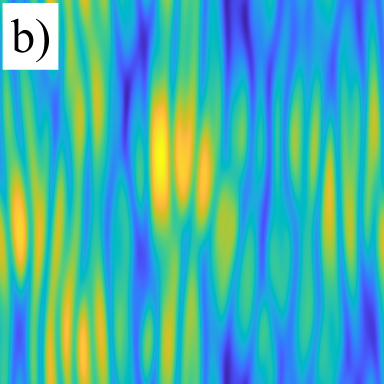} \hfill
  \includegraphics[width=.24\textwidth]{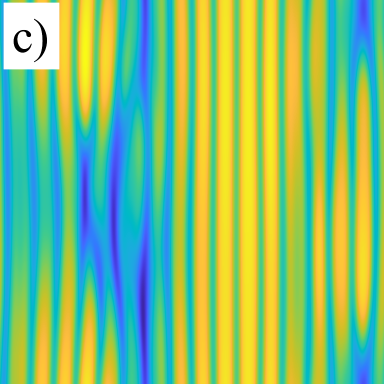} \hfill
  \includegraphics[width=.24\textwidth]{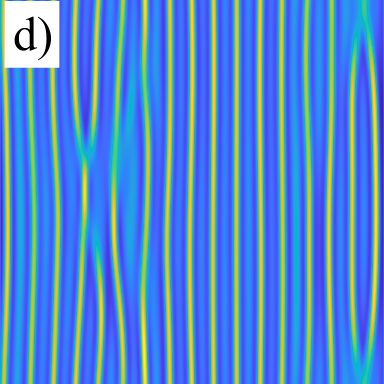}
  \caption[Comparison of ACGLE and Codimension-One Phase Equation]{A solution $A(x,y,t)$ to the ACGLE with parameters \(\alpha_1 = (-\delta-1)/\beta \approx -0.22\), \(\alpha_2 = 1/\beta = 0.20\), \(\beta=5\), \(\delta = 0.3\), and \(\kappa = -1\) was simulated on a square domain \([-200,200]\times[-200,200]\) to time \(t = 30/\delta^4 \approx 3700\). (a) Absolute value \(|A|\) and (b) phase \texttt{angle(A)}. (c) The solution \(\Phi\) of the phase equation (\ref{eq:2D_aCGL_phase_codim1}) using the same parameters at time \(T=30\). The solution \(\Phi\) displays similar behavior as the phase of the ACGLE, but with fewer pinches due to the effect of the nonlinear term and how it affects the pinching length \cite{rost1995anisotropic}. The computational domain is larger than in Figures \ref{fig:ACGLE_AvsPhi_Isotropic} and \ref{fig:ACGLE_AvsPhi_Anisotropic_codim2} because the smaller domain resulted in zero pinches and perfectly vertical ripples. (d) $|A|$ for $A$ approximated from \(\Phi\) using Equations (\ref{eq:aCGL_ansatz_longwave}), (\ref{eq:aCGL_phase_scalings_phi}) and (\ref{eq:Wapprox_func_of_Phi}).}
  \label{fig:ACGLE_AvsPhi_Anisotropic_codim1}
\end{figure}

For constructing approximative solutions of the ACGLE from solutions to the codimension-one phase equation we choose to retain the parameters in Equation (\ref{eq:2D_aCGL_phase_codim1}) so that its solutions may be directly compared to solutions of the ACGLE with the same parameters; see Figure \ref{fig:ACGLE_AvsPhi_Anisotropic_codim1}. 

The rescaled codimension-one phase equation (\ref{eq:2D_aCGL_phase_codim1_rescaled}) coincides with the equation derived by  Rost and Krug \cite{rost1995anisotropic}  who describe it as an ``elastically coupled chain" of one-dimensional Kuramoto-Sivashinsky systems. They derive an equation similar to (\ref{eq:2D_aCGL_phase_codim1_rescaled}) to analyze the ``pinching length" of the patterns produced by their version of the aKS equation, which is only anisotropic in the 2nd-order derivative terms and the nonlinear terms. They consider the case for which the 4th-order spatial derivatives are isotropic.

\subsection{Codimension-two bifurcation}
\label{subsec:cod2}

For the codimension-two case, both of the second-order derivative terms in (\ref{eq:2D_aCGL_phase_full}) become unstable simultaneously. We unfold this degeneracy, which occurs when $1+\alpha_1\beta$ and $1+\alpha_2\beta$ are both zero, by setting $1+\alpha_j\beta=\kappa_j\delta^2$ for $j=1,2$. Using $\alpha_j=-1/\beta+{\mathcal O}(\delta^2)$ in the fourth-order derivative terms and the nonlinear terms of (\ref{eq:2D_aCGL_phase_full}) and truncating this equation at ${\mathcal O}(1)$ then gives the codimension-two phase equation
\begin{equation}\label{eq:2D_aCGL_phase_codim2}
	\partT{\Phi} = \kappa_1\partXX{\Phi} + \kappa_2\partYY{\Phi} - \bigg(1+\frac{1}{\beta^2}\bigg)\frac{\partXXXX{\Phi} + 2\partXXYY{\Phi} + \partYYYY{\Phi}}{2\mu} + (\beta+\frac{1}{\beta})\big((\partX\Phi)^2 + (\partY\Phi)^2\big).
\end{equation}
This equation is an anisotropic version of the (isotropic) Kuramoto-Sivashinsky (KS) equation (\ref{eq:KSeq_Intro}) with parameters related to the ACGLE. The anisotropy is revealed in the second-order derivative terms, while the fourth-order derivative terms and the nonlinear terms are still isotropic since the codimension-two degeneracy occurs with $\alpha_1=\alpha_2$.  As a consequence, so called cancellation modes leading to blow-up solutions are not possible for (\ref{eq:2D_aCGL_phase_codim2}), in contrast to the more general anisotropic KS equation studied in \cite{kalogirou2015} and \cite{rost1995anisotropic}, in which the two nonlinear terms may have different signs.

To explore the \((\kappa_1,\kappa_2)\)-parameter plane numerically, we set $(\kappa_1,\kappa_2)=\rho(\cos\theta,\sin\theta)$. Assuming $\rho>0$, the rescaling 
$$\hat{T} =\rho\hat{\mu}T,\;\;\;(\hat{X},\hat{Y})=\sqrt{\hat{\mu}}(X,Y),\;\;\;\hat{\Phi}=\frac{2}{\rho}\big(\beta+\frac{1}{\beta}\big)\Phi,$$
with $\hat{\mu}=\frac{2\mu\rho}{1+1/\beta^2}$, simplifies (\ref{eq:2D_aCGL_phase_codim2}) to (with the hats omitted)
\begin{equation}\label{eq:2D_aCGL_phase_codim2_rescaled}
\partT\Phi=\cos(\theta)\partXX{\Phi}+\sin(\theta)\partYY{\Phi}-\big(\partXXXX{\Phi}+2\partXXYY{\Phi}+\partYYYY{\Phi}\big)+\frac{1}{2}\big((\partX\Phi)^2 + (\partY\Phi)^2\big).
\end{equation}
Two examples of numerical solutions of the rescaled codimension two phase equation (\ref{eq:2D_aCGL_phase_codim2}) are depicted in Figure \ref{fig:CGLE_codim2_soln_examples} where both \(\theta\) values are in the third quadrant so that \(\kappa_1,\kappa_2\) are both negative.

The codimension-two phase equation can also be used to recreate solutions of the ACGLE. As in the codimension-one case, we choose to retain the parameters in Equation (\ref{eq:2D_aCGL_phase_codim2}) for that purpose so that its solutions may be directly compared to solutions of the ACGLE with the same parameters. An example of this is shown in Figure \ref{fig:ACGLE_AvsPhi_Anisotropic_codim2}. The phase of the solution \(A\) to the ACGLE and the solution \(\Phi\) to the phase equation display qualitatively the same behavior. Using the lowest order approximation of \(W\) together with \(\Phi\) yields a recreation of a ACGLE solution \(A\). Using the phase equation gives considerable computational time savings. The ACGLE was simulated to time \(t\approx 5000\) while the phase equation only to time \(T=40\). Similar times savings appear in the other phase equations, as seen in Figures \ref{fig:ACGLE_AvsPhi_Isotropic} and \ref{fig:ACGLE_AvsPhi_Anisotropic_codim1}. Time savings would be much greater for smaller \(\delta\) since \(t = T/\delta^4\).  

\begin{figure}[!htbp] 
  \includegraphics[width=.45\textwidth]{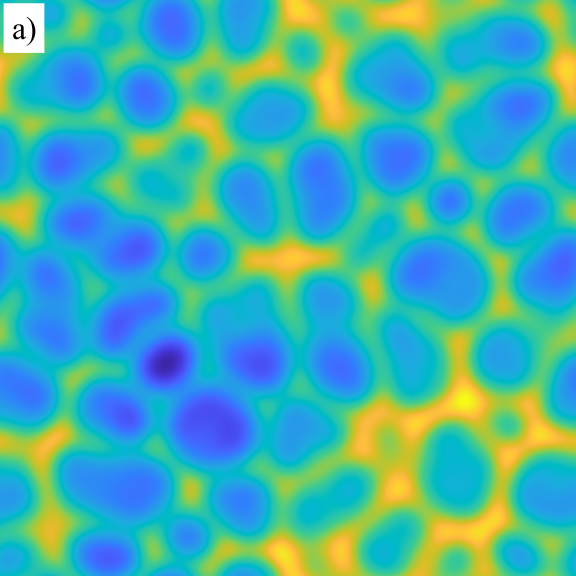} \hfill
  \includegraphics[width=.45\textwidth]{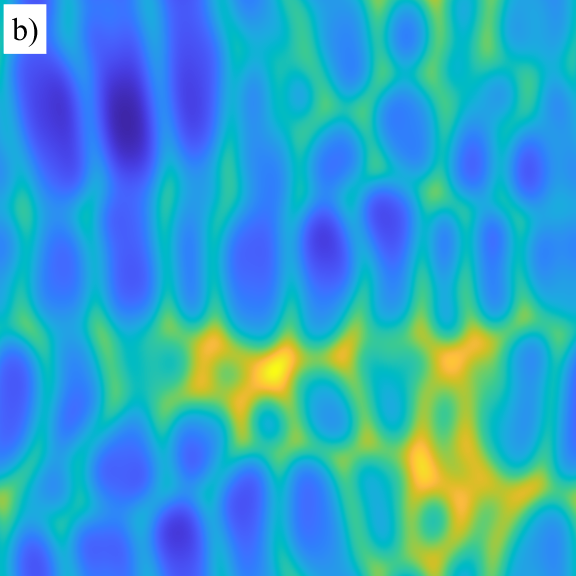}
  \caption[Codimension-two phase equation solution examples]{Snapshots of \(\Phi\) from (\ref{eq:2D_aCGL_phase_codim2}) for (a) an isotropic case with \(\theta = 1.25\pi\) and (b) an anisotropic case with \(\theta = 1.05\pi\). Both have \(\beta = 3\). The angle \(1.25\pi\) is in the isotropic region with both second order derivative terms equally unstable ($\cos\theta=\sin\theta\approx-.71$). When \(\theta = 1.05\pi\), there is still instability in both second order terms, but since $(\cos\theta,\sin\theta)\approx (-.99,-1.6)$ the instability in the \(y\)-direction is much weaker than that in the \(x\)-direction leading to the phase chaotic cells becoming elongated along the $y$-direction. If this direction would be made sufficiently stable, we would see phase-chaotic ripples instead of stretched cells.}
  \label{fig:CGLE_codim2_soln_examples}
\end{figure}

\begin{figure}[!htbp] 
  \includegraphics[width=.24\textwidth]{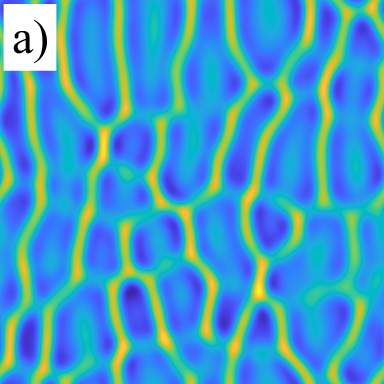} \hfill
  \includegraphics[width=.24\textwidth]{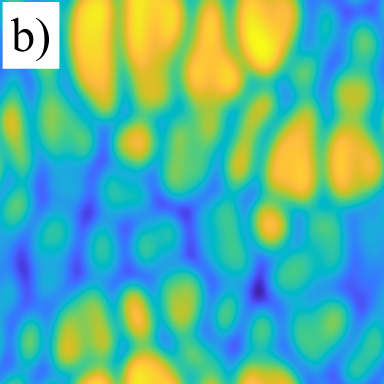} \hfill
  \includegraphics[width=.24\textwidth]{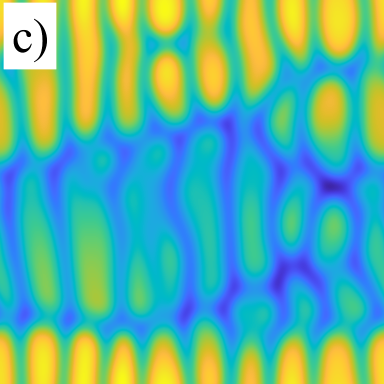} \hfill
  \includegraphics[width=.24\textwidth]{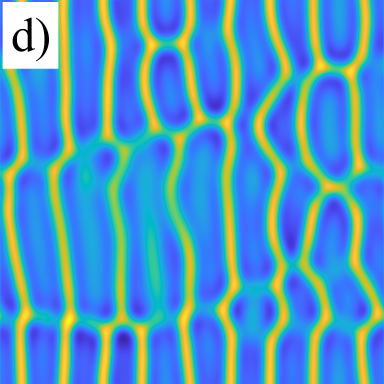}
  \caption[Comparison of ACGLE and Codimension Two Phase Equation]{A solution $A(x,y,t)$ to the ACGLE with parameters \(\beta=2\) and \(\theta = 1.05\pi\), \(\rho = 1\), \(\delta = 0.3\) corresponding to $\alpha_1=(\delta^2\rho\cos\theta-1)/\beta \approx -0.5444$, $\alpha_2 = (\delta^2\rho\sin\theta-1)/\beta \approx -0.5070$ was simulated on a square domain \([-100,100]\times[-100,100]\) to time \(t = 40/\delta^4 \approx 5000\). (a) Absolute value \(|A|\) and (b) phase \texttt{angle(A)}. (c) The solution \(\Phi\) of the phase equation (\ref{eq:2D_aCGL_phase_codim2}) using the same parameters at time \(T=40\). The solution \(\Phi\) displays similar behavior as the phase of the ACGLE, albeit not exact due to the size of \(\delta\). (d) $|A|$ for $A$ approximated from \(\Phi\) using Equations (\ref{eq:aCGL_ansatz_longwave}), (\ref{eq:aCGL_phase_scalings_phi}) and (\ref{eq:Wapprox_func_of_Phi}). The solution to \(\Phi\) exhibits fewer pinches than the solution to \(A\) at the given times, however the snapshots in (b,c) become closer to (a,b) if run to longer time. Similarly, earlier times of (a,b) more closely resemble (b,c). The effective strength of the nonlinearity, dependent on \(\beta\), and which causes the pinches appear, is not the same for the given parameters.}
  \label{fig:ACGLE_AvsPhi_Anisotropic_codim2}
\end{figure}

\section{Conclusions}
\label{sec:conclusions}

We have derived and studied Kuruamoto-Sivashinsky (KS)-type phase equations that govern the dynamics of the anisotropic complex Ginzburg-Landau equation near Benjamin-Feir-Newell (BFN) instabilities. While in the isotropic case there is just one type of BFN instability determined by $1+\alpha\beta=0$ ($\alpha=\alpha_1=\alpha_2$), the anisotropy induces two instability surfaces defined by $1+\alpha_1\beta=0$ and $1+\alpha_2\beta=0$, respectively. 

If only one of these equations is satisfied, an instability (or bifurcation) of codimension one is encountered. In this case the resulting phase equation governing the Ginzburg-Landau dynamics near such an instability is the 1D KS equation with respect to one of the two directions, with an additional diffusion term in the other direction (equation (\ref{eq:2D_aCGL_phase_codim1}) with its rescaled version (\ref{eq:2D_aCGL_phase_codim1_rescaled})). 

If both equations above are satisfied, an instability of codimension two occurs and at the instbility the anisotropic Ginzburg-Landau equation degenerates to its isotropic version. In this case the resulting phase equation is the anisotropic, two-dimensional KS equation (\ref{eq:2D_aCGL_phase_codim2}) with its rescaled version (\ref{eq:2D_aCGL_phase_codim2_rescaled}), whose anisotropy is revealed in the second order derivative terms and quantified by the angle $\theta$ while the other terms are isotropic. 

Attempts to generate solutions of the Ginzburg-Landau equation from solutions of these phase equations were successful in both cases. A paper on a systematic parameter study of the solutions of the anisotropic complex Ginzburg-Landau equation is in preparation, including parameters away from the BNF-instability surfaces.

From a general pattern formation point of view, the anisotropic complex Ginzburg-Landau equation is the generic amplitude equation for oscillatory (Hopf) instabilities with zero wavenumbers in anisotropic extended systems with reflection symmetries in both directions, \textit{i.e.} the basic state of the system becomes neutrally stable with respect to bulk oscillations. There are three types of instabilities with nonzero wavenumbers (one stationary and two oscillatory; see \cite{dangelmayr2008modulational}) for which a system of two or four coupled Ginzburg-Landau equations becomes the generic system of amplitude equations. For these coupled Ginzburg-Landau equations, BNF-type instabilities as well as Eckhaus-type instabilities result in coupled phase equations. Specifically, if the instability of the basic state is oscillatory (Hopf-type), the Ginzburg-Landau system contains global coupling terms \cite{dangelmayr2008modulational} leading to global coupling terms in the resulting coupled phase equations. These coupled phase equations are the subject of current studies.

All simulations were computed with the authors' own codes based on \cite{kassam2005fourth}, which can be found at \url{https://github.com/drhandwerk/ACGLE-Phase-Equations}.

\section*{Acknowledgements} This work was supported at Colorado State University by NSF grant DMS-1615909 to I. Oprea, G. Dangelmayr, and P. D. Shipman.

\bibliographystyle{siam}          
\bibliography{ACGLE_Phase_bib}

\end{document}